\documentclass[11pt]{article}
\usepackage{graphicx}
\usepackage{amsfonts}
\usepackage{amsmath}
\usepackage{amssymb}
\usepackage{url}
\usepackage{fancyhdr}
\usepackage{indentfirst}
\usepackage{enumerate}
\usepackage{epstopdf}
\usepackage{dsfont}
\usepackage[colorlinks=true,citecolor=blue]{hyperref}
\usepackage{amsthm}
\usepackage{color}
\usepackage{natbib}
\usepackage{longtable}
\usepackage{comment}
\usepackage{capt-of}
\usepackage{multirow}
\usepackage{subfig}
\def\d{\mathrm{d}}

\newcommand{\E}{\mathbb{E}}

\newcommand{\F}{\mathcal{F}}

\newcommand{\p}{\mathbb{P}}

\newcommand{\id}{\mathds{1}}

\renewcommand{\(}{\left(}
\renewcommand{\)}{\right)}
\renewcommand{\[}{\left[}
\renewcommand{\]}{\right]}
\renewcommand{\ge}{\geqslant}
\renewcommand{\le}{\leqslant}
\renewcommand{\geq}{\geqslant}
\renewcommand{\leq}{\leqslant}
\renewcommand{\epsilon}{\varepsilon}

\theoremstyle{plain}
\newtheorem{theorem}{Theorem}

\newtheorem{proposition}{Proposition}
\theoremstyle{definition}

\theoremstyle{remark}
\newtheorem{remark}{Remark}

\newcommand{\cet}{\begin{center}}
\newcommand{\ecet}{\end{center}}

\usepackage{setspace}

\begin{document}
	
\title{Optimal VPPI strategy under Omega ratio with stochastic benchmark}
\date{}


\author{ {\small Guohui Guan}\thanks{\scriptsize { \tiny School of Statistics, Renmin University of China, Beijing, China.
			Email: \texttt{guangh@ruc.edu.cn}}}	
	\and
	{\small Lin He}\thanks{\scriptsize  { \tiny School of Finance, Renmin University of China, Beijing, China.
			Email: \texttt{helin@ruc.edu.cn}}}
	\and
	{\small Zongxia Liang}\thanks{\scriptsize  { \tiny Department of Mathematical Sciences, Tsinghua University, China.
			Email: \texttt{liangzongxia@mail.tsinghua.edu.cn}}}
	\and
	{\small Litian Zhang}\thanks{\scriptsize { \tiny Corresponding Author.\!\! Department of Mathematical Sciences, Tsinghua University, China.\!\!
			Email: \texttt{zhanglit19@mails.tsinghua.edu.cn}}}}
\numberwithin{equation}{section}

\maketitle

\begin{abstract}	
	
	
	This paper studies a variable proportion portfolio insurance (VPPI) strategy. The objective is to determine the risk multiplier by maximizing the extended Omega ratio of the investor's cushion, using a binary stochastic benchmark. When the stock index declines, investors aim to maintain the minimum guarantee. Conversely, when the stock index rises, investors seek to track some excess returns. The optimization problem involves the combination of a non-concave objective function with a stochastic benchmark, which is effectively solved based on the stochastic version of concavification technique. We derive semi-analytical solutions for the optimal risk multiplier, and the value functions are categorized into three distinct cases. Intriguingly, the classification criteria are determined by the relationship between the optimal risky multiplier in \cite{zieling2014performance} and the value of   $1$. Simulation results confirm the effectiveness of the VPPI strategy when applied to real market data calibrations. 
	
\end{abstract}

\section{Introduction}

Risk-averse investors adopt portfolio insurance strategies to manage downside risk while capturing part of the stock market flourish.  The original idea, introduced by \cite{leland1980should}, is known as the option-based portfolio insurance (OBPI) strategy.  Subsequently, \cite{black1987simplifying} proposed the constant proportion portfolio insurance (CPPI) strategy, which not only matches the performance of OBPI but also offers a more straightforward implementation.  CPPI involves dynamically adjusting the asset allocation in the portfolio based on a pre-determined formula, which typically involves a combination of risky assets (such as stocks) and safer assets (such as bonds or cash). Researchers and practitioners in the field of asset management have extensively studied and applied CPPI, with numerous studies delving into various aspects of dynamic portfolio insurance. For instance, refer to \cite{bertrand2001portfolio, chen2008dynamic, schied2014model, hambardzumyan2019dynamic}, among others.

The risk multiplier plays a pivotal role in CPPI strategies, influencing both portfolio leverage and performance. The proportion invested in the risky asset is determined by multiplying the cushion with the risk multiplier. Traditionally, the risk multiplier has been considered fixed and predetermined. However, during periods of heightened market volatility, a substantial risk multiplier can lead to significant losses. Consequently, incorporating a time-varying risk multiplier can enhance performance. This flexibility allows investors to adjust their risk exposure in response to evolving market conditions. As a result, variable proportion portfolio insurance (VPPI) strategies have garnered extensive attention, falling into two categories: parametric optimization and non-parametric optimization.  In the context of parametric optimization, \cite{zieling2014performance} introduced an optimization model for the risk multiplier for investors exhibiting hyperbolic absolute risk aversion (HARA) preference.  
The research findings suggest that the risk multiplier is influenced by the Sharpe ratio of the risky asset and the risk-averse attitude of the investor. \cite{schied2014model} concerned a dynamic extension of the CPPI strategy in which the multiplier level may depend on
quantities including time and price evolution.   Similarly, \cite{bertrand2016portfolio} employed extreme value theory to establish an upper bound for the risk multiplier.  \cite{dichtl2011portfolio} explored the incorporation of cumulative prospect theory to capture the heterogeneous attitudes of investors towards gain and loss within the VPPI framework. \cite{ameur2014portfolio} assumed that the risk multiplier is determined
from market and portfolio dynamics. 
For non-parametric optimization, \cite{cesari2003benchmarking} integrated technical predictions, such as the 10-week moving average, into the strategy design. Additionally, statistical and machine learning methods have been employed in the design of the risk multiplier see \cite{chen2008dynamic, lee2008new, dehghanpour2017robust}, etc. All of these researches have demonstrated improved performance of VPPI strategies compared to standard CPPI strategies. This paper adopts a parametric modeling approach, aiming to identify the optimal risk multiplier based on appropriate optimization objectives.

This paper aims to determine the optimal VPPI strategy based on a specific criterion. In the literature, various approaches have been used to optimize utility or maximize returns for a given level of risk. In \cite{zieling2014performance}, the HARA utility function is optimized to find the optimal risk multiplier for CPPI strategies. Another commonly applied criterion in economics is maximizing the expected return while considering a certain level of risk. The Sharpe ratio, which quantifies the trade-off between risk and reward using variance and mean, has been widely used for this purpose. However, it has limitations as it assumes an elliptical distribution of portfolio returns. To address the limitations of the Sharpe ratio, alternative risk measures have been proposed. One such measure is the Omega ratio, introduced in \cite{OmegaRatio}, which considers the entire return distribution and distinguishes gains and losses from a specified threshold. Unlike the Sharpe ratio, the Omega ratio incorporates both upside and downside deviations and has characteristics of symmetric and asymmetric risk measures. Moreover, the Omega ratio does not assume a specific distribution of portfolio returns and is commonly applied in finance, see  \cite{kane2009optimizing, bertrand2011omega,  guastaroba2016linear, biedova2020multiplier}, etc. In this paper, we depart from the approach of \cite{zieling2014performance} and consider maximizing the Omega ratio of the terminal cushion as the objective for investors. However, it is worth noting that \cite{2019OptimalB} states that the maximization problem of the Omega ratio is unbounded in the continuous setting. To address this issue, we extend the Omega ratio following the approach in \cite{POWPR} to overcome the problem of ill-posedness. 

The Omega ratio is a measure that distinguishes gains and losses based on a specified benchmark. However, selecting an appropriate benchmark is not formally guided, as mentioned in \cite{OmegaRatio}. In the study conducted by \cite{kapsos2014worst}, it was highlighted that the solution to maximize the Omega ratio is highly sensitive to the chosen benchmark. This poses a challenge as financial markets are inherently uncertain and prone to volatility. A deterministic threshold may not adequately capture the randomness and variability of market conditions. To address this issue, researchers have explored the concept of a stochastic benchmark in several studies, such as \cite{de2011loss}, \cite{sprenger2015endowment}, and \cite{he2022endogenization}, etc. By introducing stochasticity into the benchmark, investors can better account for the unpredictable nature of asset returns and ensure a more realistic assessment of gains and losses. In the field of asset management, investors typically aim for upside capture and downside protection, often relying on a binary benchmark. In other words, when the stock index performs poorly, investors seek to maintain a minimum guarantee. Conversely, when the stock index performs well, investors aim to track some excess returns. Taking these practical considerations into account, we formulate a binary stochastic benchmark to determine the gains and losses in the extended Omega ratio.

The objective function in this context is a non-linear ratio, lacking convexity or concavity. Moreover, the optimization problem is complicated because of stochastic benchmark. These two characteristics combine to create the main challenges faced when attempting to solve this optimization problem. {Following the  fractional programming method in  \cite{dinkelbach1967nonlinear}, we transform the original optimization problem into a family of solvable linearized problems in which the objective function is defined as the numerator of the original problem minus the denominator multiplied by a penalty parameter.  Remarkably, one of these reformulated problems provides an exact solution to the original problem. Additionally, we employ the martingale method to simplify the problem, transitioning from optimizing over a space of stochastic processes to optimizing over a space of random variables, as exemplified in \cite{karatzas1998methods}. For objective functions that are not concave, concavification techniques are commonly employed, see \cite{carpenter2000does, POWPR, he2018profit}. In our paper, we tackle the optimization problem that combines a non-concave objective function with a stochastic benchmark.
 {This kind of problem lacks systematic research, and our previous work \cite{liang2021framework} investigates an abstract framework regarding this kind of problem and propose a stochastic version of the concavification technique.} Applying the results from this work, we derive the semi-analytical optimal risky multiplier and its corresponding value function.} Theoretical results indicate that both the guarantee ratio and the capturing ratio exert significant influences on the value functions. When both of these ratios are low, the value function approaches infinity. Under less stringent conditions, the cushion consistently exceeds the benchmark, leaving only room for gains. Conversely, with stricter requirements, the form of the value functions falls into one of three distinct categories. Remarkably, the criteria for this classification are based on the relationship between the optimal risky multiplier as detailed in \cite{zieling2014performance} and the value of 1. Furthermore, the optimal risky multiplier exhibits a distinct hump-shaped pattern attributed to the inherent loss aversion attitude of investors. This behavior leads investors to increase their allocation to risky assets as they approach the maturity date, driven by the desire to minimize potential losses and maximize gains through speculative investments. Notably, the optimal risk multiplier tends to be lower compared to the one under the CRRA utility framework. This difference can be attributed to the presence of loss aversion, which promotes a more cautious approach characterized by lower leverage.

While we have successfully derived semi-closed-form solutions within a continuous-time framework, it is imperative to assess the efficacy of these strategies in a discrete-time context. The effectiveness of these strategies in discrete-time scenarios is crucial, as it directly influences their practical applicability (as discussed in \cite{black1987simplifying} and \cite{balder2009effectiveness}). Specifically, our investment strategy involves adjustments at fixed time intervals, as seen in \cite{branger2010optimal} and \cite{happersberger2020estimating}, or it may be triggered by reaching specific thresholds, similar to the approach in \cite{weng2013constant}. Given that we distinguish between the minimal guarantee and the benchmark, selecting a consistent threshold can be challenging. Therefore, in this paper, we opt for the fixed time interval adjustment mechanism. For comparative analysis, we utilize the fixed multiplier strategy from \cite{zieling2014performance} as well as a constant multiplier as benchmarks. Under these settings, we evaluate the effectiveness of the strategy using various criteria, including Omega ratio, guarantee probability, expected terminal excess value, probability of underperformance, and winning probability. These metrics allow for a comprehensive assessment of the strategy's performance. Our results confirm the effectiveness of the VPPI strategy when applied to real market data calibrations. Furthermore, we observe a negative correlation between the performance of the VPPI strategy and both the required guarantee ratio and the capturing ratio. A higher guarantee (or capturing) ratio reduces the flexibility in investment and increases the difficulty in maintaining the minimal guarantee or tracking the benchmark. It is worth noting that the impact of performance parameters is multifaceted, as they not only affect investment performance but also influence variations in benchmark performance.

The contribution of the paper is threefold. First, we introduce the binary stochastic benchmark and the Omega ratio within the VPPI model, effectively capturing the objectives of investors in the field of asset management. Second, we expand the concavification technique into a stochastic framework, successfully addressing optimization challenges associated with non-concave objective functions and stochastic benchmarks. We establish semi-analytical value functions, categorizing them into three distinct cases. Lastly, our research yields several intriguing findings. We observe that the optimal risk multiplier follows a distinctive hump-shaped pattern and tends to be lower in comparison to fixed multipliers. Furthermore, the VPPI strategy outperforms alternative policies by demonstrating superior tracking of benchmarks and reducing negative gaps. We also elucidate the intricate relationship between expected returns on risky investments and the optimal strategies.

The remainder of the paper is organized as follows. Section 2 provides the mathematical formulation of the investor's optimization problem with Omega ratio and stochastic benchmark. In Section 3, we use the stochastic version of concavification techniques to solve the optimization problem. Section 4 exhibits the numerical results and  the main findings. Finally,  we conclude the paper in Section 5.

\section{Model Setting}


Consider a filtered probability space $(\Omega, \mathcal{F}, \{\mathcal{F}_t\}_{t \geq 0}, \mathbb{P})$, where $\mathbb{P}$ is the reference probability measure. Before introducing the VPPI  strategy, let us discuss the construction of the CPPI strategy.

{The current value of a portfolio at time $t$ is denoted as $V_t$}, which comprises both the value of the risky asset and the risk-free asset. The floor value, denoted as $F_t$, represents the minimum guarantee or downside protection level. The objective of the investor is to ensure that the portfolio value remains above this guarantee. The cushion at time $t$, denoted as $C_t$, is defined as the difference between the current portfolio value and the floor value, i.e.,
$$C_t=V_t-F_t.$$
In the CPPI  strategy, the risk multiplier, denoted as $m$, is a constant value. The allocation to the risky assets is determined as a proportion of the cushion. Specifically, the investment in risky assets  at time $t$ can be calculated by
\begin{equation}\label{def of C}
	e_t=m(V_t-F_t)\triangleq mC_t.
\end{equation}
The remaining portion $(V_t - e_t)$ is allocated to risk-free assets. The CPPI strategy is a straightforward approach that has found wide application in asset allocation. The risk multiplier $m$ plays a crucial role in the CPPI strategy and has been assigned different values in previous research. For instance, $m$ was set to 5 in \cite{herold2007total}, 14 in \cite{annaert2009performance}, and 3 in \cite{zhang2015optimal}. Specifically,  \cite{biedova2020multiplier} chose the optimal multiplier value with respect to the Omega ratio in the CPPI strategy. 

However, it is important to note that using a constant risk multiplier can have limitations. One such limitation is a lack of flexibility in adapting to changing market conditions. Market volatility and risk levels can fluctuate over time, and a fixed multiplier may not effectively respond to these variations. Additionally, a constant risk multiplier can lead to larger losses during periods of significant market downturns. Setting the multiplier too high could result in a higher allocation to risky assets when market conditions rapidly deteriorate, potentially amplifying losses.

The VPPI strategy addresses the limitations of a constant risk multiplier by employing time-varying multipliers that adjust dynamically based on market conditions. In the VPPI strategy, the constant risk multiplier $m$ is replaced with a square integrable and adaptive process denoted as $$\{m_t, t\ge 0\}.$$ The value of $m_t$ changes over time according to a predefined rule or algorithm, allowing it to capture changes in market conditions and risk levels. In  \cite{zieling2014performance}, the optimal risk multiplier for the VPPI strategy is obtained by maximizing the CRRA utility of the cushion. This approach takes into account investors' risk preferences and aims to find the risk multiplier that optimizes their utility function. In our work, different from \cite{zieling2014performance}, {we optimize the risk multiplier for the VPPI strategy using the extended Omega ratio. The extended Omega ratio is a risk-adjusted performance measure that considers both downside risk and upside potential. By optimizing the risk multiplier based on this measure, we aim to enhance the risk-adjusted performance of the VPPI strategy.}

Similar to \cite{biedova2020multiplier}, we consider continuous time rebalancing in the Black-Scholes market. There are two assets in the financial market: a  risk-free asset $B$ and a  risky asset $S$, which are described as follows:
\begin{equation*}
	\left\{
	\begin{aligned}
		\frac{dB_t}{B_t}&=rdt,~B_0=1,\\
		\frac{dS_t}{S_t}&=\mu dt+\sigma dW_t,~S_0=1,
	\end{aligned}
	\right.
\end{equation*}
where $\{W_t: t \geq 0\}$ is a standard Brownian motion defined on the filtered probability space  $(\Omega, \F, \{\F_t\}_{t \geq 0}, \p)$. The first equation represents the risk-free asset, with the constant interest rate $r$ representing the risk-free rate of return. The second equation describes the risky asset, with $\mu$ representing the expected return and $\sigma$ representing the volatility of the asset.

Suppose that the investor has an initial value $V_0=1$, and the guarantee $F_t$ in the VPPI strategy satisfies the equation
\begin{equation*}
	dF_t=rF_tdt, ~F_0=kV_0=k,
\end{equation*}
where $k$ is the guarantee proportion.  The initial guarantee $F_0$ is set as a fixed proportion $k$ of the initial portfolio value $V_0$. This setup ensures that the guarantee keeps pace with the growth of the risk-free asset and provides a minimum protection level for the portfolio.

In the Omega ratio, as originally proposed by \cite{OmegaRatio}, the benchmark is traditionally considered as a constant value that distinguishes gains from losses. The benchmark used by the investor  is often stochastic and varies according to the prevailing market conditions. Specifically, when the financial market performs well, the investor wants to track a portion of returns exceeding the index. When it performs poorly, the benchmark is set to be 0, i.e.,  the investor wants to maintain the minimal guarantee. To evaluate the cushion $C_T$ of the investor, we introduce a binary benchmark formulated as follows:
\begin{equation}\label{def of Y}
	Y=
	\left\{
	\begin{aligned}
		&0,&S_T<F_T,\\
		&\eta (S_T-F_T),&S_T\ge F_T.
	\end{aligned}
	\right.
\end{equation}
Here $\eta\le 1$ is a multiplier. Our goal is to find the optimal $m$ to maximizing the extended Omega ratio with a binary benchmark:
\begin{equation}\label{Omega ratio}
	\max\limits_{\{m_t,0\le t\le T\}}\frac{\E\left[U(C_T-Y)\id_{\{C_T>Y\}}\right]}{\E\left[U(Y-C_T)\id_{\{C_T\le Y\}}\right]},
\end{equation}
where $U(x)=x^{\gamma}, \gamma<1$. In Problem \eqref{Omega ratio}, when $C_T$ exceeds $Y$ ($C_T$ falls short of $Y$), the investor realizes gains (suffers losses). The investor's preference is characterized by the attitude towards gains relative to the attitude towards losses. When the parameter $\gamma$ is set to 1, Problem \eqref{Omega ratio} represents the standard Omega ratio. However, as noted by \cite{2019OptimalB}, in the continuous setting, the maximization problem of the Omega ratio becomes unbounded. To address this, we extend the Omega ratio in Problem \eqref{Omega ratio} by incorporating the function $U(x) = x^{\gamma}$. This extension has also been explored in \cite{POWPR}.  

\section{Optimal multiplier}
In this section, we focus on deriving the optimal multiplier for Problem \eqref{Omega ratio}. This problem involves fractional optimization with a stochastic benchmark. Stochastic dynamic programming method cannot be applied. To address it, we initially employ the martingale method, as outlined in \cite{KLSX1991}, and the linearization technique, as described in \cite{POWPR}. These methods allow us to transform Problem \eqref{Omega ratio} into a non-concave optimization problem within a space of random variables featuring a stochastic benchmark. Subsequently, we employ the stochastic version of the concavification technique, as proposed in \cite{liang2021framework}, to tackle and resolve the transformed problem.
\subsection{Martingale method}
We begin by deriving the dynamics of the cushion, represented as the excess value $C_t$. This dynamic is a consequence of \eqref{def of C}, and it follows that
\begin{equation}\label{eq of C}
\begin{split}
	dC_t&=dV_t-dF_t
	=(V_t-e_t)\frac{dB_t}{B_t}+e_t\frac{dS_t}{S_t}-dF_t\\
	&=(C_t+F_t-m_tC_t)rdt+m_tC_t\frac{dS_t}{S_t}-rF_tdt\\
	&=(C_t-m_tC_t)rdt+m_tC_t\frac{dS_t}{S_t}\\
	&=\left[m_t(\mu-r)+r\right]C_tdt+m_t\sigma C_tdW_t,
\end{split}
\end{equation}
and $C_0=V_0-kV_0=1-k$.
Based on \eqref{eq of C}, we can employ the martingale method, as described in \cite{KLSX1991} and \cite{CH1989}, to transform the original problem into the following optimization problem concerning the terminal cushion:
\begin{equation}\label{reduced problem}
	\sup_{C_T\in\mathcal{C}}~\frac{\E\left[U(C_T-Y)\id_{\{C_T>Y\}}\right]}{\E\left[U(Y-C_T)\id_{\{C_T\le Y\}}\right]}
\end{equation}
with budget constraint
\begin{equation}\label{def of flower C}
	\mathcal{C}=\left\{C_T~\text{is nonnegative and }\mathcal{F}_T\text{-measurable}~\bigg|~\mathbb{E}\left[\xi_T C_T\right]\le 1-k\right\},
\end{equation}
where $\xi_t=\exp\left\{-rt-\frac{\theta^2}{2}t-\theta W_t\right\}$ is the pricing kernel with $0\le t\le T$  and $\theta = \frac{\mu -r}{\sigma}$. 
\subsection{Linearization}
The objective function in Problem \eqref{reduced problem} is  a fractional form, which cannot be solved directly. To address this, we utilize fractional programming theory and introduce a family of linearized problems. This family of problems is parameterized by a non-negative factor $\lambda$, as described in \cite{POWPR}:
\begin{equation}\label{def of f}
f(\lambda) = \sup_{C_T\in\mathcal{C}}~\left\{\E\left[U(C_T-Y)\id_{\{C_T>Y\}}\right]-\lambda\E\left[U(Y-C_T)\id_{\{C_T\le Y\}}\right]\right\}.
\end{equation}
By linearizing the original fractional objective function, we transform the problem into a set of linearized problems. We have the following proposition on relationship between $f(\cdot)$ in Problem \eqref{def of f} and the optimal solution to Problem \eqref{reduced problem}:
\begin{proposition}\label{prop linearized}
Assume $1-k<\E[\xi_T Y]$, then $f(\cdot)$ has a zero point $\lambda^*$. If the optimization problem \eqref{def of f} with $\lambda$ being $\lambda^*$ has an optimal solution $C_T^*$, then $C_T^*$ is also an optimal solution for Problem \eqref{reduced problem}, and $\lambda^*$ is the  unique  value  function of Problem \eqref{reduced problem}.
\end{proposition}
\begin{proof}
	See Appendix \ref{app1}.
\end{proof}
\begin{remark}
	Note that if $1-k\ge\E[\xi_T Y]$, then one can find a $C_T\in\mathcal{C}$ such that $C_T\ge Y$ a.s., which leads to a zero term in the denominator of Problem \eqref{reduced problem}, and the problem is meaningless. This indicates that the benchmark $Y$ should not be set too small relative to the guarantee level $k$.
\end{remark}
\subsection{Solution to the linearized probelm}
Proposition \ref{prop linearized} suggests that to solve the original problem \eqref{reduced problem}, it is sufficient to find solutions to the linearized problem \eqref{def of f}.
Rewrite Problem \eqref{def of f} as
\begin{equation}\label{redef of f}
f(\lambda) = \sup_{C_T\in\mathcal{C}}~\E\left[u_{\lambda}(C_T,Y)\right],
\end{equation}
where $u_{\lambda}$ is defined as
\begin{equation*}
	u_{\lambda}(c,y)=\left\{
	\begin{aligned}
		&(c-y)^{\gamma}, &&c>y,\\
		&-\lambda(y-c)^{\gamma},&&c\le y.
	\end{aligned}
	\right.
\end{equation*}
To address the non-concave optimization problem  \eqref{redef of f} with a stochastic benchmark, we follow the approach outlined in \cite{liang2021framework} and introduce the following function:
\begin{equation*}
	\mathcal{X}_{\lambda}(z,y)=\mathop{\arg\sup}_{c\ge0}{\big[u_{\lambda}(c,y)-cz\big]}.
\end{equation*}
We perform some computational procedures to obtain the expression for $\mathcal{X}_{\lambda}(z,y)$ as follows:
\begin{equation*}
\mathcal{X}_{\lambda}(z,y)=
\left\{
\begin{aligned}
	&~~0,~&&z>f_{\lambda}(y),\\
	&\left(\frac{z}{\gamma}\right)^{\frac{1}{\gamma-1}}+y,~&&z<f_{\lambda}(y),\\
	&\left\{0,\left(\frac{z}{\gamma}\right)^{\frac{1}{\gamma-1}}+y\right\},~&&z=f_{\lambda}(y),
\end{aligned}
\right.
\end{equation*}
where $f_{\lambda}(y)$ is the unique solution of the following equation:
\begin{equation}\label{eq1}
	(1-\gamma)\left(\frac{x}{\gamma}\right)^{\frac{\gamma}{\gamma-1}}-yx+\lambda y^{\gamma}=0.
\end{equation}
In the following context, unless stated otherwise, we consider $\mathcal{X}_{\lambda}(z,y)=0$ at the point $z=f_{\lambda}(y)$. This assumption reduces $\mathcal{X}_{\lambda}$ to a single-valued function. Drawing on the findings presented in \cite{liang2021framework}, we can solve Problem \eqref{redef of f} based on  the following theorem:
\begin{theorem}\label{thm of solution}
	For every $\lambda\ge0$, Problem \eqref{redef of f} has a unique optimal solution
	\begin{equation*}
		C^*_T=\mathcal{X}_{\lambda}(\nu^*\xi_T,Y)=
		\left[\left(\frac{\nu^*\xi_T}{\gamma}\right)^{\frac{1}{\gamma-1}}+Y\right]\id_{\{\nu^*\xi_T<f_{\lambda}(Y)\}},
	\end{equation*}
	where the real number $\nu^*>0$  solves the budget constraint $\E\[\xi_T\mathcal{X}_{\lambda}(\nu^*\xi_T,Y)\]=1-k$.
\end{theorem}
\begin{proof}
	See Appendix \ref{app1}.
\end{proof}

Upon solving the optimal terminal cushion $C_T^*$, by the martingale method, the optimal cushion $C_t^*$  at time $t$ is also  expressed as follows:
\begin{equation}\label{Ct}
	C^*_t=\frac{1}{\xi_t}\E\left[\xi_T C^*_T\big|\mathcal{F}_t\right]
	=\E\left[Z_t \left(\left(\frac{\nu^*\xi_tZ_t}{\gamma}\right)^{\frac{1}{\gamma-1}}+Y\right)\id_{\{\nu^*\xi_tZ_t<f_{\lambda}(Y)\}}\bigg|\mathcal{F}_t\right].
\end{equation}
Based on results above, we formulate the procedure of solving Problem \eqref{reduced problem} as follows: for every $\lambda\ge0$, we solve Problem \eqref{redef of f} by finding a constant $\nu^*(\lambda)$ depends on $\lambda$ such that $\E\[\xi\mathcal{X}_{\lambda}(\nu^*(\lambda)\xi,Y) \]=1-k$, and then get the value of $f(\lambda)$. Once we find the zero point $\lambda^*$ of $f$, the original Problem \eqref{reduced problem} is also solved by Proposition \ref{prop linearized}. The numerical results are displayed in Section \ref{sec:nume}.

The following proposition is useful when calculating the optimal strategies of Problem \eqref{redef of f}:
\begin{proposition}\label{prop exp}
	Let $Z_t=\frac{\xi_T }{\xi_t}$. Then, for $v\in\mathbb{R}$ and $0\le p\le q$, we have
	\begin{equation*}
		\begin{aligned}
			\E\left[Z_t^v\id_{\{p<Z_t<q\}}\right]=e^{v(T-t)\left[\frac{\theta^2}{2}(v-1)-r\right]}\Bigg[\Phi\left(\left(v-\frac{1}{2}\right)\theta\sqrt{T-t}-\frac{\ln p+r(T-t)}{\theta\sqrt{T-t}}\right)-\\
			\Phi\left(\left(v-\frac{1}{2}\right)\theta\sqrt{T-t}-\frac{\ln q+r(T-t)}{\theta\sqrt{T-t}}\right)\Bigg],
		\end{aligned}
	\end{equation*}
where $\Phi$ is the cumulative distribution function of the standard normal distribution.
\end{proposition}
\begin{proof}
	The proof involves straightforward calculus computations and is therefore omitted for brevity.
\end{proof}	
In the remainder of this section, our objective is to solve Problem \eqref{redef of f} and derive the expression for the optimal multiplier. To simplify the formulas, we will utilize the symbol $\phi$ to represent the probability density function of  $\Phi$. Additionally, we introduce the following notations:
\begin{equation*}
	\begin{aligned}
		H_t(x,v)\triangleq\Phi\left(\left(v-\frac{1}{2}\right)\theta\sqrt{T-t}-\frac{\ln x+r(T-t)}{\theta\sqrt{T-t}}\right),\\
		G_t(x,v)\triangleq\phi\left(\left(v-\frac{1}{2}\right)\theta\sqrt{T-t}-\frac{\ln x+r(T-t)}{\theta\sqrt{T-t}}\right).
	\end{aligned}
\end{equation*}

Based on Definition \eqref{def of Y} of $Y$, we rewrite $Y$ as
\begin{equation*}
	Y=\eta\left(a\xi_T^b-ke^{rT}\right)\id_{\{\xi_T<c\}}\triangleq g(\xi_T),
\end{equation*}
where
\begin{equation*}
	a=\text{exp}\left\{\frac{T}{2}(\mu+r)\left(1-\frac{\sigma^2}{\mu-r}\right)\right\},~b=-\frac{\sigma^2}{\mu-r}<0,~c=\left(\frac{ke^{rT}}{a}\right)^{\frac{1}{b}}.
\end{equation*}
To deal with the term $\nu^*\xi_T<f_\lambda(Y)$ in \eqref{Ct}, we apply substitution method and let $x=\gamma(sy)^{\gamma-1}$ in \eqref{eq1}, which leads to
\begin{equation*}
	(1-\gamma)y^{\gamma}s^{\gamma}-\gamma y^{\gamma}s^{\gamma-1}+\lambda y^{\gamma}=0.
\end{equation*}
Therefore, $s$ satisfies
\begin{equation*}
	(1-\gamma)s^{\gamma}-\gamma s^{\gamma-1}+\lambda =0,
\end{equation*}
which is an equation that dose not dependent on $y$, and we have $f_\lambda(y)=dy^{\gamma-1}$ with $d=\gamma s^{\gamma-1}$ being a constant.
Therefore,
$\nu^*\xi_T<f_\lambda(Y)$ is equivalent to
\begin{equation*}
	\left(\xi_T\ge c\right) \text{ or }
	\left(\xi_T< c \text{ and }
	\nu^*\xi_T<d\eta^{\gamma-1}\left(a\xi_T^b-ke^{rT}\right)^{\gamma-1}\right).
\end{equation*}
To calculate the second part {in the condition above}, we define
$h(z)=d_1z^{\frac{1}{\gamma-1}}-az^b+ke^{rT}$ with
$d_1=\frac{1}{\eta}\left(\frac{\nu^*}{d}\right)^{\frac{1}{\gamma-1}}$, {and then the condition can be rewritten as $h(\xi_T)>0$.
We have the following three cases of the function $h(\cdot)$.}

\begin{enumerate}
	\item If $1-\gamma>\frac{\mu-r}{\sigma^2}$.
	In this case, we list properties of $h(\cdot)$ after trivial calculations (see Figure \ref{fig1} for instance):
\begin{equation*}
	h(0+)=-\infty,~h'(0+)=+\infty, ~h(+\infty)>0, \text{ and }
	h'(z)<(\text{ or}>)0 \text{ for } z>(\text{ or}<)z_1,
\end{equation*}
where $z_1$ is a positive real number.

\begin{figure}[htbp]
	\centering
	\begin{minipage}{0.5\textwidth}
		\centering
		\includegraphics[totalheight=6cm]{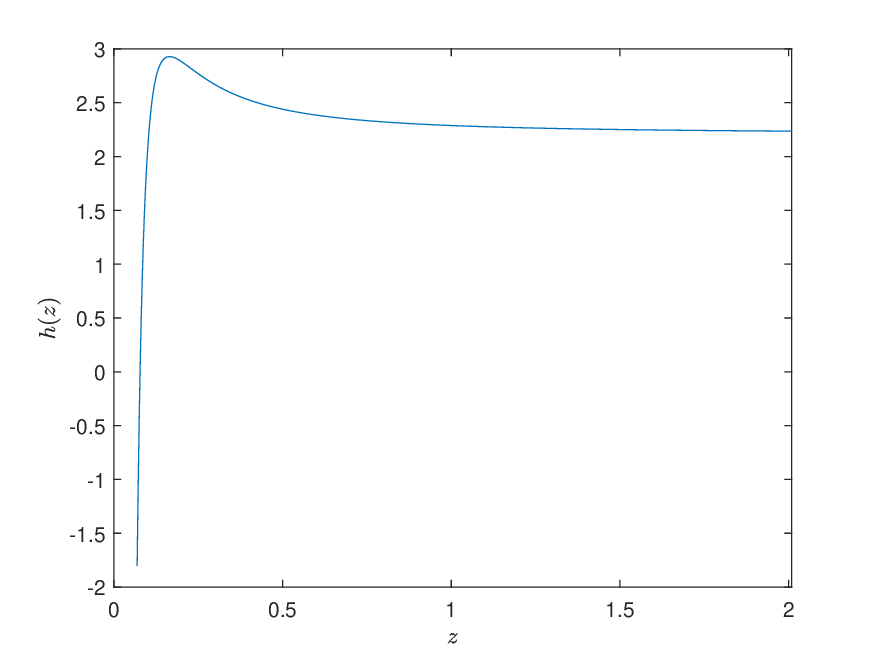}
		\caption{\mbox{Example graph of $h(z)$ when $1-\gamma>\frac{\mu-r}{\sigma^2}$.}}
		\label{fig1}
	\end{minipage}\hfill
\end{figure}
Based on above results, $h(z)>0$ leads to $z>v_1$ for some positive real number $v_1$. 
Therefore, $\nu^*\xi_T<f_\lambda(Y)$ is equivalent to $\xi_T > v_1$, and Formula \eqref{Ct} becomes
\begin{equation}\label{Ct2}
	C^*_t
	=\E\left[Z_t \left(\left(\frac{\nu^*\xi_tZ_t}{\gamma}\right)^{\frac{1}{\gamma-1}}+g(\xi_tZ_t)\right)\id_{\{\xi_tZ_t>v_1\}}\bigg|\mathcal{F}_t\right]=n(\xi_t,t),
\end{equation}
where
\begin{equation*}
\begin{aligned}
\!\!\!	n(z,t)=&\E\left[Z_t \left(\left(\frac{\nu^*zZ_t}{\gamma}\right)^{\frac{1}{\gamma-1}}+g(zZ_t)\right)\id_{\{zZ_t>v_1\}}\right]\\
	=&\left(\frac{\gamma}{\nu^*}\right)^{\frac{1}{1-\gamma}}z^\frac{1}{\gamma-1}\E\left[Z_t^{\frac{\gamma}{\gamma-1}}\id_{\{Z_t>\frac{v_1}{z}\}}\right]\\
	& +a\eta z^{b}\E\left[Z_t^{b+1}\id_{\{\frac{v_1}{z}<Z_t<\frac{c}{z}\}}\right]\!-\!k\eta e^{rT}\E\left[Z_t\id_{\{\frac{v_1}{z}<Z_t<\frac{c}{z}\}}\right].
\end{aligned}
\end{equation*}

Using Proposition \ref{prop exp} yields
\begin{equation}\label{n1}
	\begin{aligned}
	n(z,t) =& D_1z^\frac{1}{\gamma-1}H_t\(\frac{v_1}{z},\frac{\gamma}{\gamma-1}\)
	+D_2z^b\left[H_t\(\frac{v_1}{z},b+1\)-H_t\(\frac{c}{z},b+1\)\right]\\
	&-k\eta e^{rt}\left[H_t\(\frac{v_1}{z},1\)-H_t\(\frac{c}{z},1\)\right],
	\end{aligned}
\end{equation}
where $D_1=\left(\frac{\gamma}{\nu^*}\right)^{\frac{1}{1-\gamma}}e^{\frac{\gamma}{\gamma-1}(T-t)\left(\frac{\theta^2}{2\gamma-2}-r\right)}$,
$D_2=a\eta e^{-\frac{\mu+r}{2}(b+1)(T-t)}$.

To derive the optimal strategy $\{m_t^*,0\le t\le T\}$ from the optimal $C_t^*$ for Problem \eqref{redef of f}, using It\^{o}'s formula we obtain
$$
\d C_t^*=\frac{\partial n}{\partial z}(\xi_t,t)\d \xi_t
+\frac{\partial n}{\partial t}(\xi_t,t)\d t
+\frac{\theta^2}{2}\frac{\partial^2 n}{\partial z^2}(\xi_t,t)\xi_t^2\d t.
$$
Comparing the last equation with 
Eq.\eqref{eq of C}, we have
\begin{equation}\label{eq m*}
m^*_t = -\frac{\theta\xi_t}{\sigma C_t^*}\frac{\partial n}{\partial z}(\xi_t,t)
=-\frac{\theta}{\sigma C^*_t}M(\xi_t,t),	
\end{equation}

where $M(z,t)=z\frac{\partial n}{\partial z}(z,t)$ and has the following expression in this case:
\begin{equation}\label{mt1}
	\begin{aligned}
		M(z,t)
		=&\frac{D_1}{\gamma-1}z^\frac{1}{\gamma-1}H_t\(\frac{v_1}{z},\frac{\gamma}{\gamma-1}\)
		+bD_2z^{b}\left[H_t\(\frac{v_1}{z},b+1\)-H_t\(\frac{c}{z},b+1\)\right]\\	
		&+\frac{D_1z^\frac{1}{\gamma-1}}{\theta\sqrt{T-t}}G_t\(\frac{v_1}{z},\frac{\gamma}{\gamma-1}\)
		+\frac{D_2z^{b}}{\theta\sqrt{T-t}}\left[G_t\(\frac{v_1}{z},b+1\)-G_t\(\frac{c}{z},b+1\)\right]\\
		&-\frac{k\eta e^{rt}}{\theta\sqrt{T-t}}\left[G_t\(\frac{v_1}{z},1\)-G_t\(\frac{c}{z},1\)\right].
	\end{aligned}
\end{equation}
\item If $1-\gamma<\frac{\mu-r}{\sigma^2}$.
In this case, we can verify the following properties of $h$ (see Figures \ref{fig2} \& \ref{fig3} for examples):
\begin{equation*}
	h(0+)=+\infty,~h'(0+)=-\infty, ~h(+\infty)>0, \text{ and }
	h'(z)>0 \text{ for } z>z_2,
\end{equation*}
{where $z_2$ is a positive real number.} 

Based on above properties of $h$, the condition $\nu^*\xi_T<f_\lambda(Y)$ finally leads to two more cases:
\begin{enumerate}[(i)]
	\item $\xi_T>0$;
	\item $\xi_T\in(0,v_2)\cup(v_3,+\infty)$, {where $v_2$ and $v_3$ are some real numbers satisfying $0<v_2<v_3<c$.}
\end{enumerate}
\begin{figure}[htbp]
	\centering
	\begin{minipage}{0.5\textwidth}
		\centering
		\includegraphics[totalheight=6cm]{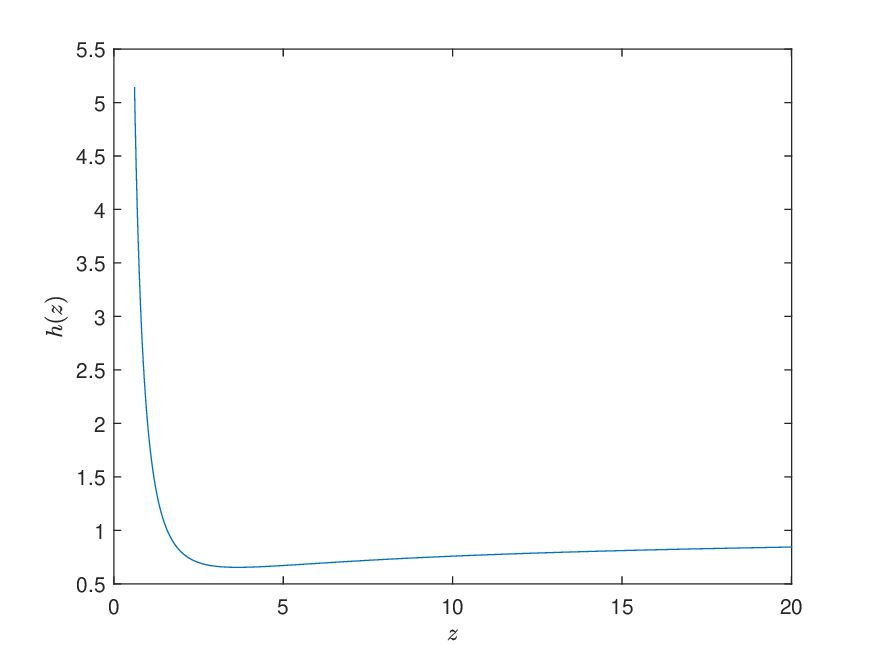}
		\caption{Example graph of $h$ for Case 2-(i).}
		\label{fig2}
	\end{minipage}\hfill
	\begin{minipage}{0.5\textwidth}
	\centering
	\includegraphics[totalheight=6cm]{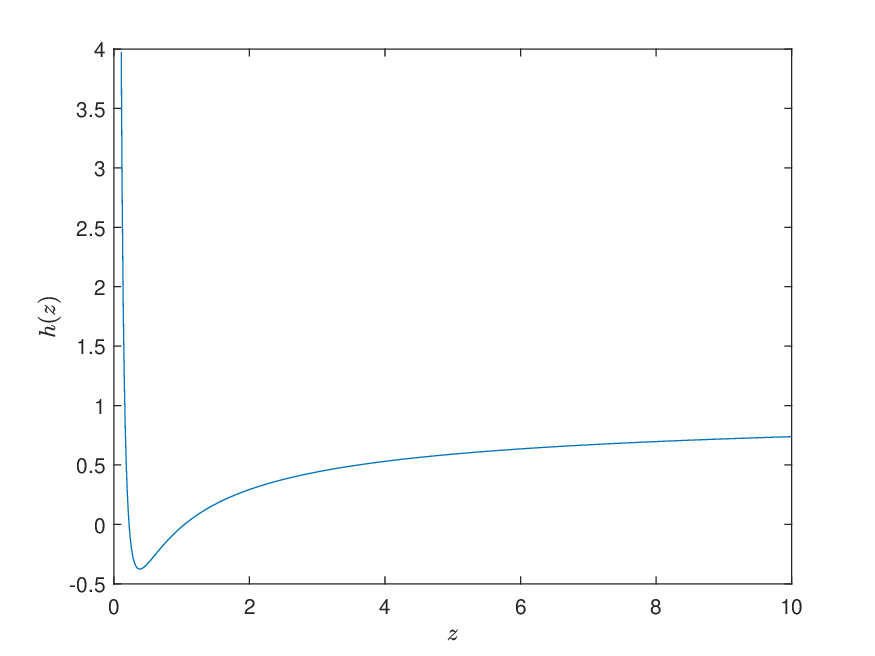}
	\caption{Example graph of $h$ for Case 2-(ii).}
	\label{fig3}
\end{minipage}\hfill
\end{figure}
Using Proposition \ref{prop exp}, we can find out the expression of $C_t^*=n(\xi_t,t)$ and $m^*_t
=-\frac{\theta}{\sigma C^*_t}M(\xi_t,t)$.

For Case 2-(i), $n$ and $M$ can be obtained by letting $v_1\rightarrow0^+$ in \eqref{n1} and \eqref{mt1}, and we have
\begin{equation}\label{n2}
	\begin{aligned}
	n(z,t) = D_1z^\frac{1}{\gamma-1}
+D_2z^b\left[1-H_t\(\frac{c}{z},b+1\)\right]
-k\eta e^{rt}\left[1-H_t\(\frac{c}{z},1\)\right],
	\end{aligned}
\end{equation}
\begin{equation}\label{mt2}
\begin{aligned}
	M(z,t)
	=&\frac{D_1}{\gamma-1}z^\frac{1}{\gamma-1}
	+bD_2z^{b}\left[1-H_t\(\frac{c}{z},b+1\)\right]	
	-\frac{D_2z^{b}}{\theta\sqrt{T-t}}G_t\(\frac{c}{z},b+1\)\\
	&+\frac{k\eta e^{rt}}{\theta\sqrt{T-t}}G_t\(\frac{c}{z},1\).
\end{aligned}
\end{equation}
For Case 2-(ii), we have
\begin{equation*}
\begin{aligned}
n(z,t)
=&D_1z^\frac{1}{\gamma-1}\left[H_t\(\frac{v_3}{z},\frac{\gamma}{\gamma-1}\)-H_t\(\frac{v_2}{z},\frac{\gamma}{\gamma-1}\)+1\right]\\
&+D_2 z^{b}\left[H_t\(\frac{v_3}{z},b+1\)-H_t\(\frac{c}{z},b+1\)
-H_t\(\frac{v_2}{z},b+1\)+1\right]\\
&-k\eta e^{rt}\left[H_t\(\frac{v_3}{z},1\)-H_t\(\frac{c}{z},1\)
-H_t\(\frac{v_2}{z},1\)+1\right],
\end{aligned}
\end{equation*}
\begin{equation}\label{mt3}
	\begin{aligned}
		M(z,t)
		=&\frac{D_1z^\frac{1}{\gamma-1}}{\gamma-1}\left[H_t\(\frac{v_3}{z},\frac{\gamma}{\gamma-1}\)-H_t\(\frac{v_2}{z},\frac{\gamma}{\gamma-1}\)+1\right]\\
		&+bD_2 z^{b}\left[H_t\(\frac{v_3}{z},b+1\)-H_t\(\frac{c}{z},b+1\)
		-H_t\(\frac{v_2}{z},b+1\)+1\right]\\
		&+\frac{D_1z^\frac{1}{\gamma-1}}{\theta\sqrt{T-t}}
		\left[G_t\(\frac{v_3}{z},\frac{\gamma}{\gamma-1}\)-G_t\(\frac{v_2}{z},\frac{\gamma}{\gamma-1}\)\right]\\
		&+\frac{D_2 z^{b}}{\theta\sqrt{T-t}}
		\left[G_t\(\frac{v_3}{z},b+1\)-G_t\(\frac{c}{z},b+1\)
		-G_t\(\frac{v_2}{z},b+1\)\right]\\
		&-\frac{k\eta e^{rt}}{\theta\sqrt{T-t}}
		\left[G_t\(\frac{v_3}{z},1\)-G_t\(\frac{c}{z},1\)
		-G_t\(\frac{v_2}{z},1\)\right].
	\end{aligned}
\end{equation}
\item If $1-\gamma=\frac{\mu-r}{\sigma^2}$,
then $h(z)=(d_1-a)z^b+ke^{rT}$, and $\nu^*\xi_T<f_\lambda(Y)$ finally leads to $\xi_T>v_4$ for some constant $v_4\in[0,c)$. The case reduces to Case 1.
\end{enumerate}

Summarizing the above derivations, we have the following theorem for the optimal risk multiplier.
\begin{theorem}
The optimal risk multiplier is given by 
$$
m^*_t =-\frac{\theta}{\sigma C^*_t}M(\xi_t,t),
$$
where $ C^*_t=n(\xi_t,t)$ and 
\begin{enumerate}
	\item when $1-\gamma\ge\frac{\mu-r}{\sigma^2}$, { $n(z,t)$ and $M(z,t)$ are given by \eqref{n1} and \eqref{mt1}}, respectively.
	\item when $1-\gamma<\frac{\mu-r}{\sigma^2}$, { $n(z,t)$ and $M(z,t)$ are given by \eqref{n2}--\eqref{mt3}}, repsectively.
\end{enumerate}
\end{theorem}

Recall that we have established an optimal solution $C_T^*$ for Problem \eqref{redef of f} for every $\lambda\ge0$. By utilizing numerical methods, we can determine the zero point $\lambda^*$ of the function $f(\cdot)$. After obtaining $\lambda^*$, we  determine the corresponding category in which it falls and employ the related expressions for $n$ and $M$ to derive the optimal multiplier process $\{m_t^*, 0 \leq t \leq T\}$. The numerical results will be presented in the following section.


%

\section{\bf Numerical results}\label{sec:nume}

In this section, we present the numerical results for Problem \eqref{Omega ratio}. For the market parameter settings, we consider $\mu=0.1435$ and $\sigma=0.17$, which are estimated using the average return of the S\&P 500 Index from 2012 to 2021. Additionally, we set $r=0.0088$ based on the US 10-year Treasury Yield during the same period. Regarding the model settings, we choose $T=5$ years, $k=0.9$, $\eta=0.7$, and $\gamma=0.5$.

We employ the Monte Carlo method to simulate both the market price, denoted as $\{S_t, 0\le t\le T\}$, and the cushion, denoted as $\{C_t, 0\le t\le T\}$, under various PPI strategies. These strategies involve fixed values of the risk multiplier $m$ selected from the set $\{2, 3, 4, 5, 6, 8, 10\}$, in addition to the optimal risk multiplier determined in our research. Furthermore, we consider the specific choice of $m=\frac{\mu-r}{(1-\gamma)\sigma^2}=9.32$, which corresponds to the optimal strategy for a CRRA utility function $u(x)=x^\gamma$. Our simulations are conducted with the assumption of 260 trading days per year. In alignment with the methodology described in \cite{zieling2014performance}, we introduce a borrowing constraint by implementing the condition $m_tC_t\le2V_t$ in our simulations. This condition ensures that whenever $m_tC_t$ exceeds $2V_t$, we restrict $m_t$ to the level of $\frac{2V_t}{C_t}$. To maintain the stability of the simulation process, particularly when the term $C_t^*$ in the denominator of \eqref{eq m*} approaches zero, we constrain the values of $m^*_t$ to fall within the interval $[0, 20]$.

\subsection{\bf Comparisons with CPPI strategies}
\begin{figure}[htbp]
	\centering
	\begin{minipage}{0.5\textwidth}
		\centering
		\includegraphics[totalheight=5.9cm]{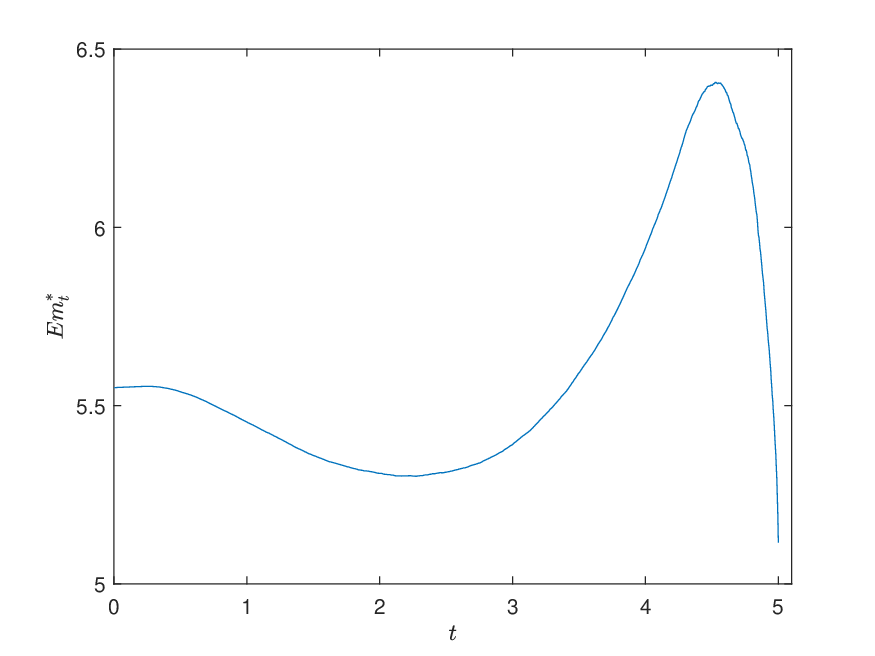}
		\caption{Mean of $m^*_t$.}
		\label{fig5}
	\end{minipage}\hfill
	\begin{minipage}{0.5\textwidth}
		\centering
		\includegraphics[totalheight=5.9cm]{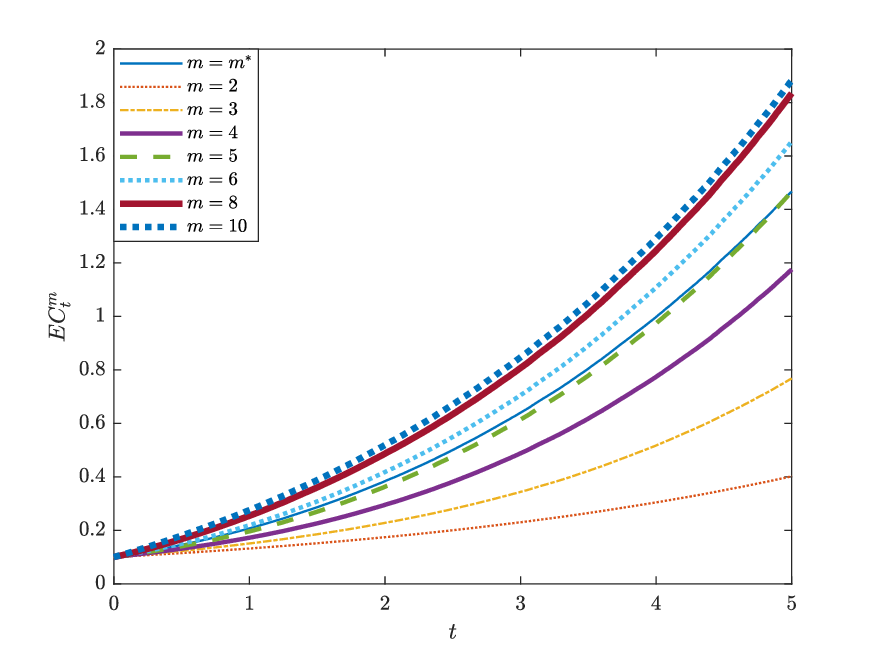}
		\caption{Mean of $C^m_t$ under different $m$.}
		\label{fig4}
	\end{minipage}\hfill
\end{figure}

In Figure \ref{fig5}, the mean of $m^*_t$ exhibits a distinct hump-shaped pattern. Similar to the outcomes observed under utility functions characterized by loss aversion, investors tend to increase their allocation to the risky asset as they approach the maturity date. This strategic shift is aimed at mitigating potential losses and capitalizing on the prospects of higher returns through a sort of investment ``gambling". Specifically, in the baseline model where the risky asset performs well, this shift towards gambling behavior is relatively gradual. Consequently, it results in the optimal $m^*_t$ falling within the range of $[5, 6.5]$. Notably, the optimal $m^*_t$ remains below the value of $m=\frac{\mu-r}{(1-\gamma)\sigma^2}=9.32$. This discrepancy can be attributed to the influence of loss aversion, which encourages a more cautious approach with lower leverage.

In Figure \ref{fig4}, we conduct a comparison of the mean values of $C^m_t$ across different $m$ settings.  The results reveal a positive correlation between the mean of $C^m_t$ and the parameter $m$. Furthermore, the optimal excess value exhibits a gradual increase over time. This trend can be attributed to the performance of the risky asset, where higher leverage results in greater returns. However, it is important to note that employing higher leverage also introduces higher downside risk. As a result, the optimal case based on the mean criterion may not necessarily align with achieving the highest utility when considering the Omega ratio.

To conduct a comprehensive comparison of performance under different $m$ values, we calculate several performance indicators for each $m$ as follows:
\begin{itemize}
	\item $E_1$: the expected utility of the numerator in the  extended Omega ratio \eqref{Omega ratio}
	\item $E_2$: the expected utility of the denominator in the extended  Omega ratio \eqref{Omega ratio}
	\item {$OR$}: the extended Omega ratio $\frac{E1}{E2}$
	\item $E[C_T]$: the expectation of the terminal cushion $C_T$
	\item $P_{liq}$: the probability of liquidation (i.e., $C_T<0$)
	\item $P$: the probability that the terminal cushion $C_T$ is less than the benchmark $Y$
\end{itemize}
Furthermore, for each constant $m$, we calculate the winning rate  of $m^*$ against the constant $m$, which represents the probability that the terminal cushion $C_T^m$ under the constant strategy $m$ is smaller than $C_T^{m^*}$. We denote this winning rate as $WR$.
\begin{longtable}{cccccccc}
	\caption{Performance results for $m^*$ and constant multiplier strategies}\\
	\hline
	 & $E_1$ & $E_2$ & $OR$ & $P_{liq}$ & $P$ & $WR$ & $E[C_T]$\\
	\hline
	\endfirsthead
	\multicolumn{8}{c}%
	{\tablename\ \thetable\ -- \textit{Continued from previous page}} \\
	\hline
	$m$ & $E_1$ & $E_2$ & $OR$ & $P_{liq}$ & $P$ & $WR$ & $EC_T$ \\
	\hline
	\endhead
	\hline \multicolumn{8}{r}{\textit{Continued on next page}} \\
	\endfoot
	\hline
	\endlastfoot
	$m^*$ & 0.505 & 0.123 & 4.11 & 0.0 & 63.3\% & - & 1.470\\
	2 & 0.014 & 0.568 & 0.02 & 0.0 & 66.8\% & 75.7\% & 0.402\\
	3 & 0.132 & 0.335 & 0.39 & 0.0 & 21.0\% & 75.4\% & 0.768\\
	4 & 0.367 & 0.230 & 1.60 & 0.0 & 40.8\% & 70.4\% & 1.175\\
	5 & 0.519 & 0.204 & 2.54 & 0.0 & 50.0\% & 44.6\% & 1.468\\
	6 & 0.605 & 0.203 & 2.98 & 0.0 & 53.6\% & 29.2\% & 1.655\\
	8 & 0.677 & 0.226 & 3.00 & 0.0 & 54.5\% & 31.4\% & 1.836\\
	10 & 0.692 & 0.261 & 2.65 & 0.0 & 52.6\% & 45.6\% & 1.884\\
	$\frac{\mu-r}{(1-\gamma)\sigma^2}$ & 0.691 & 0.249 & 2.78 & 0.0 & 53.4\% & 39.5\% & 1.877\\
	\hline
\end{longtable}

The findings highlight several notable trends. First, there is a positive correlation between the expected utility of the numerator ($E_1$) in the Omega ratio and the parameter $m$, indicating that higher leverage results in increased returns. Conversely, the expected utility of the denominator ($E_2$) in the Omega ratio exhibits an interesting pattern. Initially, it decreases as $m$ is enlarged, but then it begins to increase. The initial decrease in $E_2$ can be attributed to the reduction of insufficient leverage, which hinders effective tracking of the benchmark. The subsequent increase in $E_2$ is a consequence of excessive leverage, which can lead to higher downside risk.

Remarkably, the expected utility of the denominator in the Omega ratio is the smallest under the optimal value of $m^*$, and the probability of the terminal excess value $C_T$ falling below the benchmark $Y$ is the second highest in this case. These seemingly contradictory results can be attributed to the strategy designed to maximize the Omega ratio. Initially, investors select a moderate value for $m^*$. As the maturity date approaches, the performance of the risky asset and the high excess value of the fund incentive  investors to take on some level of risk, positively impacting the tracking of the upward benchmark. However, there is a higher likelihood of falling below the benchmark in the latter stages due to the larger leverage. Nevertheless, these later strategies can be effectively and timely adjusted, reducing the associated gap and resulting in relatively small negative utility. Furthermore, several other observations can be made. The Omega ratio is maximized at the optimal value of $m^*$, and the minimum guarantee is fully protected in all cases. Although the ideal winning rate is not achieved, the {strategy $m^*$} maximizes the Omega ratio by closely tracking the benchmarks and minimizing the gaps.

\begin{figure}[htbp]
	\centering
	\begin{minipage}{0.5\textwidth}
		\centering
		\includegraphics[totalheight=6cm]{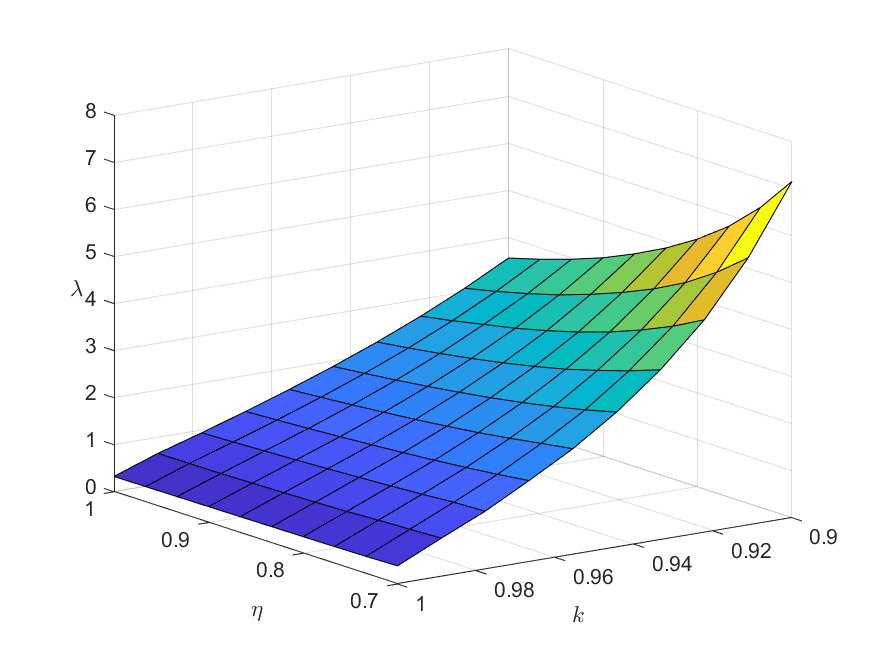}
		\caption{Theoretical optimal Omega ratio.}
		\label{fig6}
	\end{minipage}\hfill
	\begin{minipage}{0.5\textwidth}
		\centering
		\includegraphics[totalheight=6cm]{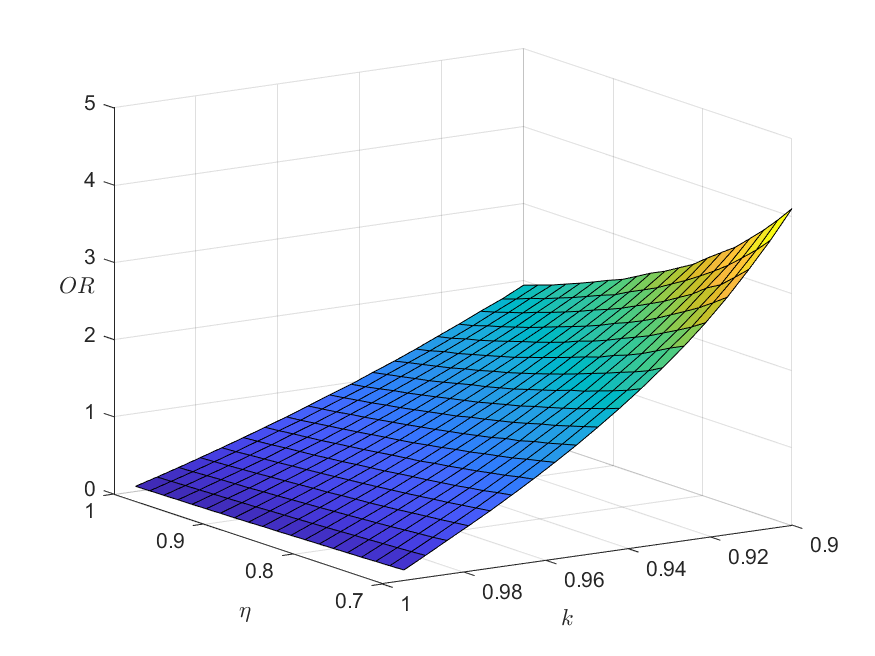}
		\caption{Realized Omega ratio.}
		\label{fig7}
	\end{minipage}
\end{figure}

Additionally, we have generated a plot illustrating the Omega ratio as it varies with different values of $\eta$ and $k$. It is important to note that in Proposition \ref{prop linearized}, a requirement is set that $1-k<\E[\xi_T Y]$, which is equivalent to $0.2\eta + k>1$ given the current parameter settings. Therefore, we have selected the intervals $k\in[0.9,1]$ and $\eta\in[0.7,1]$ as the target range where this inequality automatically holds. As our simulations are conducted in a discrete setting, the realized Omega ratio {$OR$} (Figure \ref{fig7}) may not perfectly align with the theoretical optimal Omega ratio $\lambda$ (Figure \ref{fig6}). Nevertheless, both figures exhibit a similar overall trend. It is evident that the Omega ratio exhibits a negative correlation with both $k$ and $\eta$. A higher value of $k$ increases the challenge of safeguarding the minimum guarantee and limits investment flexibility. Likewise, a larger value of $\eta$ heightens the difficulty of tracking the upward benchmark and reduces the capacity for returns.
\begin{figure}[htbp]
	\centering
	\begin{minipage}{0.5\textwidth}
		\centering
		\includegraphics[totalheight=7cm]{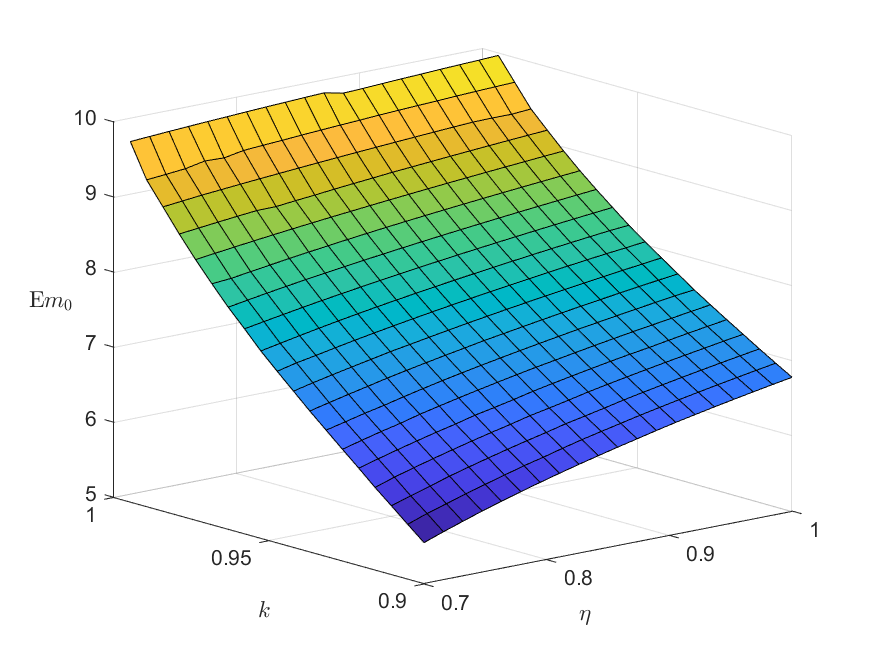}
		\caption{$\E [m_0]$.}
		\label{fig8}
	\end{minipage}\hfill	
\end{figure}

Furthermore, we depict the relationship between $\E [m_0]$ and the parameters $(\eta, k)$ in Figure \ref{fig8}. Notably, as $k$ increases, the excess value $C$ decreases in response. This reduction in excess value results in reduced investment flexibility, necessitating an increase in the expected risky proportion $\E [m_0]$ to counter the adverse effects. Similarly, when $\eta$ takes on larger values, it becomes more challenging to effectively track the upward benchmark. Consequently, the expected risky proportion $\E[ m_0]$ must also be elevated to achieve better tracking performance. In summary, the risky multiplier is positively correlated with both $k$ and $\eta$.

\subsection{\bf Sensitivity test}


\begin{figure}[htbp]
	\centering
	\subfloat[Relationshiop between $\E m_0$ and $\mu$.]{\includegraphics[width=0.5\textwidth]{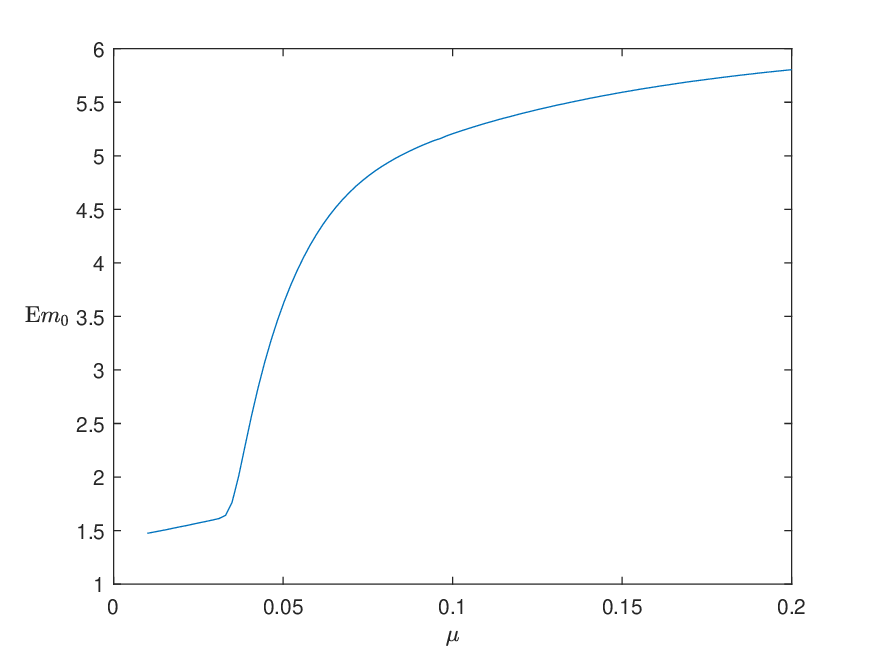}}\hfill
	\subfloat[Relationshiop between  $\lambda$ and $\mu$.]{\includegraphics[width=0.5\textwidth]{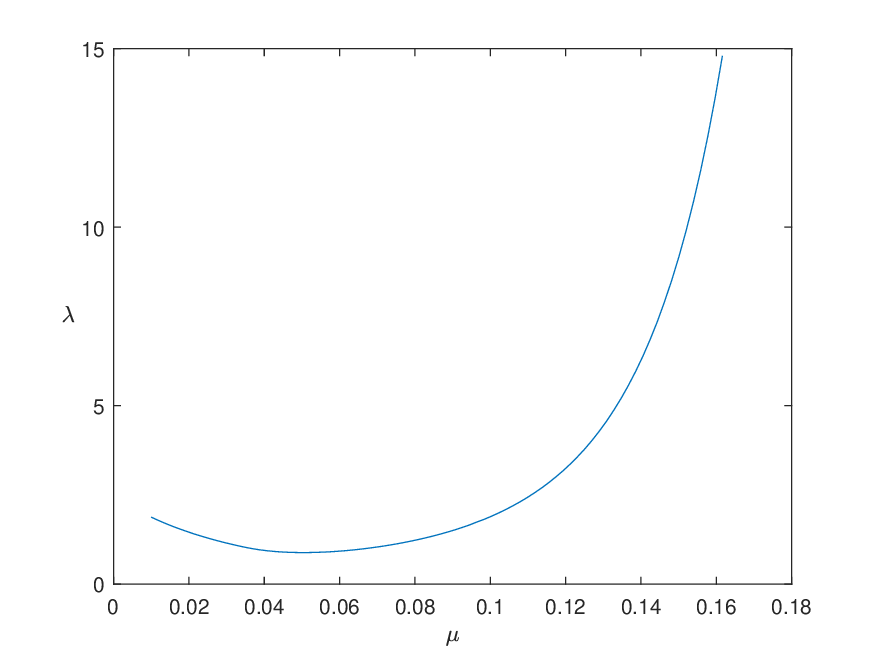}}\\
	
	\subfloat[Relationshiop between  $\E m_0$ and $\sigma$.]{\includegraphics[width=0.5\textwidth]{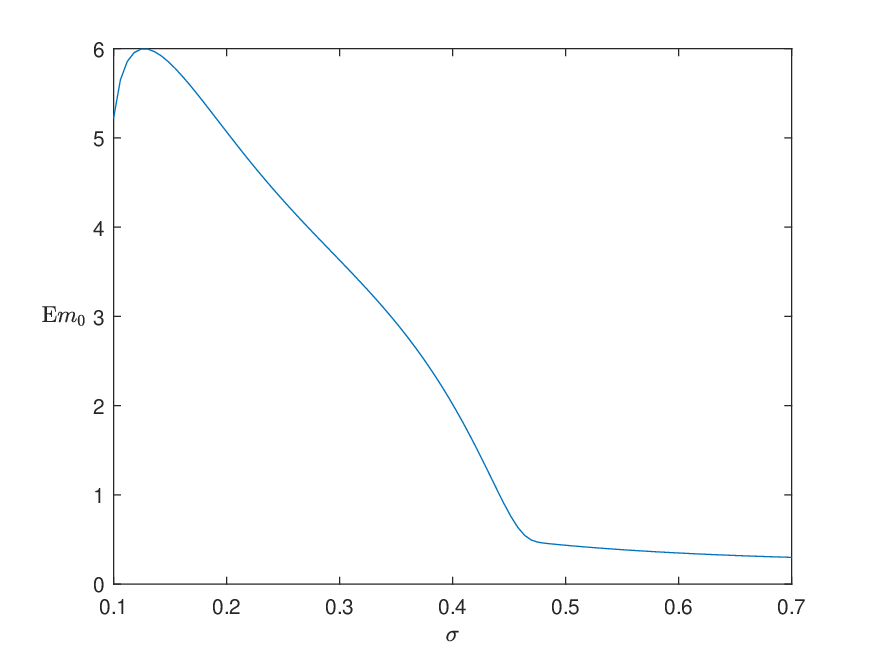}}\hfill
	\subfloat[Relationshiop between $\lambda$ and $\sigma$.]{\includegraphics[width=0.5\textwidth]{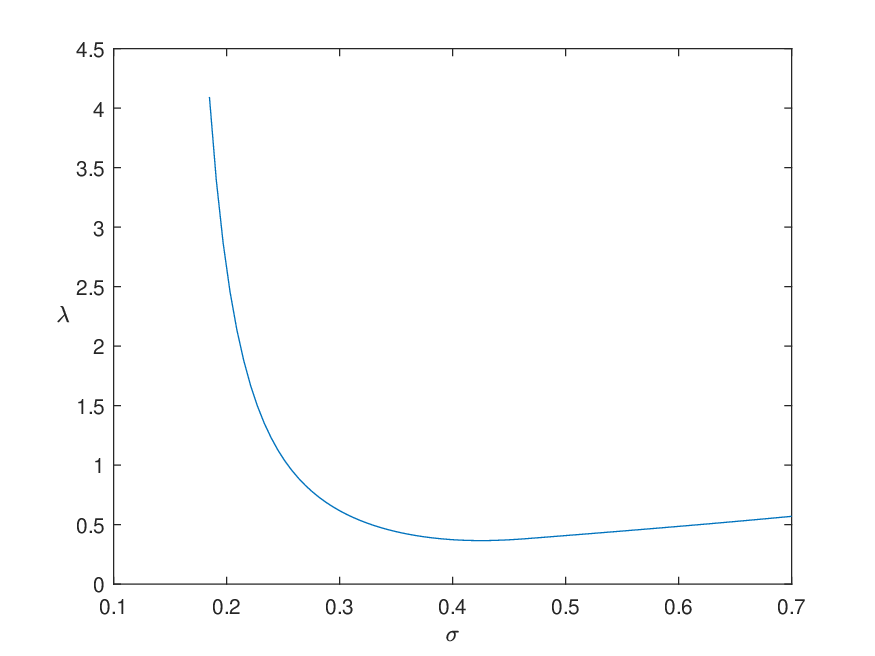}}\\
	
	
	\subfloat[Relationshiop between  $\E m_0$ and $\gamma$.]{\includegraphics[width=0.5\textwidth]{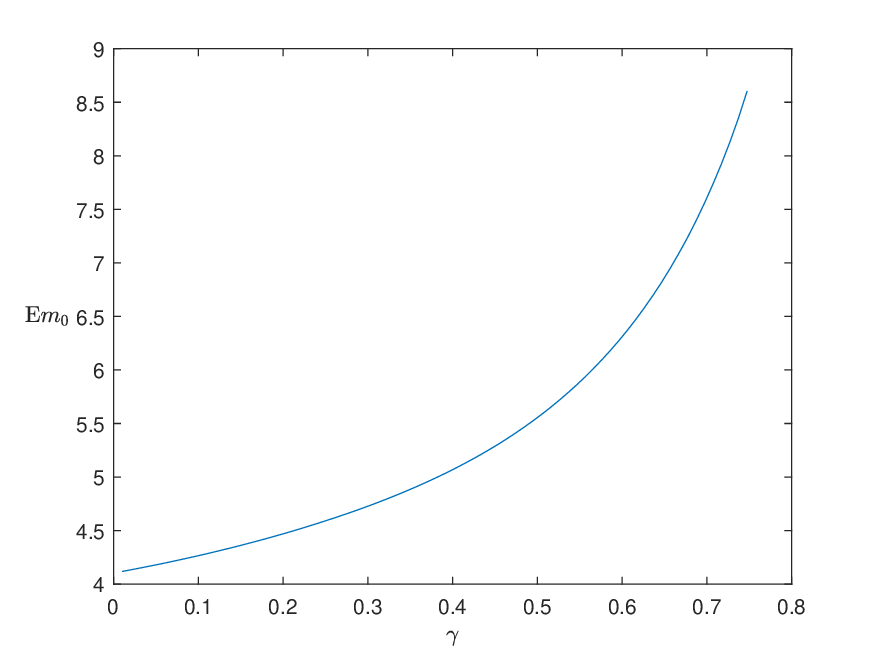}}\hfill
	\subfloat[Relationshiop between $\lambda$ and $\gamma$.]{\includegraphics[width=0.5\textwidth]{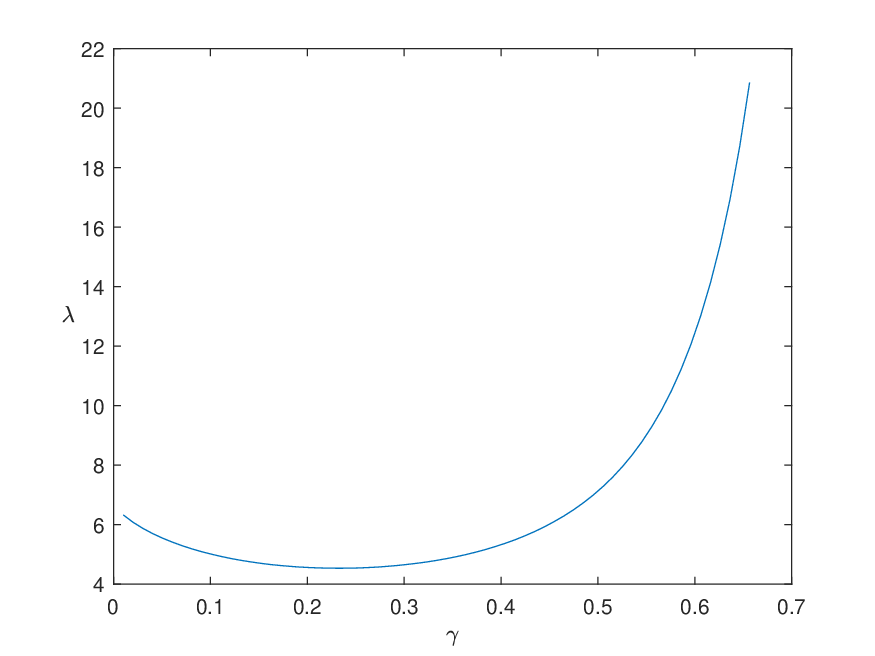}}
	
	\caption{Sensitivity analysis.}
	\label{fig:sensitivity}
\end{figure}
In this subsection,  {we present the relationships between $(OR, \E [m_0])$ and the market parameters $(\mu, \sigma, \gamma)$ in Figure \ref{fig:sensitivity}. As Figure \ref{fig6} and Figure \ref{fig7} indicate that the theoretical Omega ratio $\lambda$ and the realized Omega ratio $OR$ exhibit a similar trend, we use $\lambda$ instead of $OR$ in the following figures because the latter requires much more calculation than the former.} Figures \ref{fig:sensitivity}-(a) and \ref{fig:sensitivity}-(b) illustrate the impacts of $\mu$ on $\E [m_0]$ and $\lambda$. The influence of $\mu$ can be outlined in three key aspects. First, a higher expected return for the risky asset results in a more frequent occurrence of an upward benchmark. Second, this upward benchmark that needs to be tracked increases accordingly. Last, the actual return on investments in the risky asset also experiences an increase.
These combined effects lead to a larger proportion of the portfolio allocated to risky investments as $\mu$ increases. Furthermore, when $\mu$ increases, the Omega ratio  initially decreases and then increases. In the initial stage, the increase in $\mu$ transforms the benchmark from a downside protection level to an upward tracking benchmark. This transition elevates the challenge of achieving positive returns, thereby reducing the Omega ratio. In the latter stage, the performance of risky investment returns becomes exceptionally robust. Under such circumstances, not only is the benchmark tracked effectively, but also substantial excess returns are generated. As a result, the Omega ratio increases with the increase of $\mu$.

In Figures \ref{fig:sensitivity}-(c) and \ref{fig:sensitivity}-(d), we investigate the effects of $\sigma$ on $\E m_0$ and $\lambda$. The findings indicate that both $\E m_0$ and $\lambda$ exhibit negative correlations with $\sigma$. When $\sigma$ increases, it introduces heightened volatility in both upward and downward directions. In particular, the increased downward volatility has a substantial adverse impact on the Omega ratio, leading investors to reduce their allocation to risky assets.

Figures \ref{fig:sensitivity}-(e) and \ref{fig:sensitivity}-(f) examine the effects of $\gamma$ on $\E m_0$ and $\lambda$. The outcomes reveal that both $\E m_0$ and $\lambda$ exhibit positive correlations with $\gamma$. As $\gamma$ increases, investors exhibit a reduced level of risk aversion. In such a scenario, investors favor allocating a higher proportion to risky assets. Additionally, the weakening of loss aversion, as indicated by the increasing $\gamma$, contributes to an overall elevation in the Omega ratio.

\section{\bf Conclusion}

In this paper, we study the optimal risk multiplier within the VPPI strategy, which enables investors to adjust their risk exposure based on market conditions. Our objective is to maximize the extended Omega ratio while incorporating a binary stochastic benchmark. To address the non-concave nature of this optimization problem, we employ the stochastic version of the concavification technique, which enables us to derive semi-analytical solutions for the optimal risk multiplier. The value functions are divided into three categories, depending on the relationship between the fixed risk multiplier introduced in the previous work by \cite{zieling2014performance} and the value 1.

The findings of this study reveal that the optimal risk multiplier exhibits a hump-shaped pattern and tends to be lower compared to fixed multipliers. This suggests that investors adopt a cautious approach influenced by loss aversion. Our results highlight that, due to the influence of loss aversion, investors tend to favor lower leverage {at the beginning} and strategically pursue larger positive returns as the maturity date approaches. Interestingly, although the over-performance probability and winning rate may not reach ideal levels, the VPPI {strategy} excels in maximizing the Omega ratio through close benchmark tracking and reduced negative gaps. {Moreover, the impact of performance parameters, such as the guarantee ratio and capturing ratio, has a multifaceted effect on both investment and benchmark performance.}

\vskip 10pt\noindent
{\bf Acknowledgements} 

The authors acknowledge support from the National Natural Science Foundation of China (Grant Nos. 12371477, 11901574, 12271290, 11871036), the MOE Project of Key Research Institute of Humanities and Social Sciences (22JJD910003), and the National Social Science Fund of China (20AZD075). The authors thank the members of the group of Mathematical Finance and Actuarial Science at the Department of Mathematical Sciences, Tsinghua University for their feedback and useful conversations.

\bibliographystyle{apalike} 
\bibliography{refs}

\begin{appendix}
	
	\section{Proof of main results}
	\label{app1}
	\begin{proof}[Proof of Theorem \ref{thm of solution}]
		Define
		\begin{equation*}
			\underline{X}_{\lambda}(z,y)=\min \mathcal{X}_{\lambda}(z,y)=
			\left\{
			\begin{aligned}
				&~~0,~&&z\ge f_{\lambda}(y),\\
				&\left(\frac{z}{\gamma}\right)^{\frac{1}{\gamma-1}}+y,~&&z<f_{\lambda}(y),
			\end{aligned}
			\right.
		\end{equation*}
		and   $g(\nu)=\E\[\xi_T\underline{X}_{\lambda}\(\nu \xi_T,Y\)\]$ for $\nu>0$. As $\xi_T$ is log-normal distributed, the expectation is always well-defined, and one can verify  $g(0+)=+\infty$, $g(+\infty)=0$.
		Therefore, {Proposition 1 and Theorem 3 in \cite{liang2021framework} ensure the existence of optimal solution to Problem \eqref{redef of f} and the finiteness of Problem \eqref{redef of f} for all $k\in (0,1)$.} Moreover, the random set $\mathcal{X}_{\lambda}\(\nu \xi_T,Y\)$ is { not } single-valued if and only if $\nu \xi_T=f_{\lambda}(Y)$, which is equivalent to
		\begin{equation*}
			d\[\eta\left(a\xi_T^b-ke^{rT}\right)\id_{\{\xi_T<c\}}\]^{\gamma-1}=\nu\xi_T.
		\end{equation*}
		As the probability that the above equation holds is $0$, Theorem 3 in \cite{liang2021framework} ensures the uniqueness of the optimal solution and the continuity of $g$. Therefore, there exists $\nu^*>0$ such that $g(\nu^*)=1-k$, and again Theorem 3 ensures that $\underline{X}_{\lambda}\(\nu^* \xi_T,Y\)$ is the optimal solution to Problem \eqref{redef of f}.
	\end{proof}
	
\begin{proof}[Proof of Proposition \ref{prop linearized}]
	First, we prove that $f$ is finite. As $\xi_T$ is log-normal distributed, Proposition 1 in \cite{liang2021framework} ensures that the problem
	$\sup\limits_{C_T\in\mathcal{C}}~\E\left[U(C_T-Y)\id_{\{C_T>Y\}}\right]$ has a finite optimal value $v_1$. Therefore, we have
	$
	-\lambda \E\[U(Y)\]\le f(\lambda)\le v_1.
	$
	
	 Second, we prove that $f$ is continuous. For $0\le \lambda_1<\lambda_2$, we have

\begin{equation*}
	\begin{aligned}
		f(\lambda_1)&=\sup_{C_T\in\mathcal{C}}~\left\{\E\left[U(C_T-Y)\id_{\{C_T>Y\}}\right]-\lambda_1\E\left[U(Y-C_T)\id_{\{C_T\le Y\}}\right]\right\}\\
		&\le f(\lambda_2)+(\lambda_2-\lambda_1)\sup_{C_T\in\mathcal{C}}\E\left[U(Y-C_T)\id_{\{C_T\le Y\}}\right]\\
		&\le f(\lambda_2)+(\lambda_2-\lambda_1)\E \[U(Y)\].
	\end{aligned}
\end{equation*}
According to the definition of $f$, we know $f(\lambda_2) \le f(\lambda_1)$. Therefore, we have
$$
|f(\lambda_2)-f(\lambda_1)|\le(\lambda_2-\lambda_1)\E \[U(Y)\],
$$
which indicates that $f$ is continuous.

Third, we prove 
$
v_2\triangleq\inf\limits_{C_T\in\mathcal{C}}\E\left[U(Y-C_T)\id_{\{C_T\le Y\}}\right]>0.
$
Indeed,  Theorem 3 in \cite{liang2021framework} ensures that the problem
$
\sup\limits_{C_T\in\mathcal{C}}\E\left[-U(Y-C_T)\id_{\{C_T\le Y\}}\right]
$
admits an optimal solution $C^*$. Noting that $\E\[\xi_TC^*_T\]\le 1-k<\E\[\xi_T Y\]$, it is impossible that $C^*_T\ge Y$ holds almost everywhere. Therefore,
$
v_2 = \E\left[U(Y-C^*_T)\id_{\{C^*_T\le Y\}}\right]>0.
$

Using the results above, we obtain $f(\lambda)\le v_1-\lambda v_2$. The inequality indicates that $f(\lambda)<0$ for sufficiently large $\lambda$. Recall that $f$ is continuous and $f(0)=v_1\ge0$, the intermediate value theorem guarantees a zero point $\lambda^*$ for $f$.

For this zero point $\lambda^*$, we know 
$$
0=f(\lambda^*)=\sup_{C_T\in\mathcal{C}}~\left\{\E\left[U(C_T-Y)\id_{\{C_T>Y\}}\right]-\lambda^*\E\left[U(Y-C_T)\id_{\{C_T\le Y\}}\right]\right\}.
$$
This indicates that 
$\E\left[U(C_T-Y)\id_{\{C_T>Y\}}\right]-\lambda^*\E\left[U(Y-C_T)\id_{\{C_T\le Y\}}\right]\le0$ holds
for every $C_T\in\mathcal{C}$. Noting that $\E\[\xi_TY\]>1-k\ge \E\[\xi_TC_T\]$, we have
$\E\left[U(Y-C_T)\id_{\{C_T\le Y\}}\right]>0$. Therefore, the following inequality holds:
$$
\frac{\E\left[U(C_T-Y)\id_{\{C_T>Y\}}\right]}{\E\left[U(Y-C_T)\id_{\{C_T\le Y\}}\right]}\le\lambda^*.
$$
As $C_T^*$ is an optimal solution to the optimization problem $f(\lambda^*)$, we have
$$
0=f(\lambda^*)=\E\left[U(C_T^*-Y)\id_{\{C_T^*>Y\}}\right]-\lambda^*\E\left[U(Y-C_T^*)\id_{\{C_T^*\le Y\}}\right],
$$
and hence
$$
\frac{\E\left[U(C_T^*-Y)\id_{\{C_T^*>Y\}}\right]}{\E\left[U(Y-C_T^*)\id_{\{C_T^* \le Y\}}\right]}=\lambda^*.
$$
This result shows that $\lambda^*$ is the optimal value of Problem \eqref{reduced problem} and $C_T^*$ is an optimal solution.
\end{proof}

\end{appendix}
\end{document}